\documentclass[12pt,notitlepage]{article}%
\usepackage[utf8]{inputenc}
\usepackage{amsthm}
\usepackage{setspace}
\usepackage{eurosym}
\usepackage{subfigure}
\usepackage{xcolor}
\usepackage{amssymb}
\usepackage{graphicx}
\usepackage[colorlinks,linkcolor=links,citecolor=cites,urlcolor=MyDarkBlue]%
{hyperref}
\usepackage{amsfonts}
\usepackage{amsmath}
\usepackage{amstext}
\usepackage{appendix}
\usepackage{appendix}
\usepackage{mathpazo}
\usepackage{tikz}
\usetikzlibrary{matrix}
\usetikzlibrary{decorations.pathreplacing}
\usetikzlibrary{calc}
\usepackage{natbib}%

\providecommand{\U}[1]{\protect\rule{.1in}{.1in}}
\newtheorem{theorem}{Theorem}

\newtheorem{definition}{Definition}
\newtheorem{example}{Example}

\newtheorem{lemma}{Lemma}

\numberwithin{equation}{section}
\definecolor{MyDarkBlue}{rgb}{0,0.08,0.45}
\definecolor{cites}{HTML}{324b13}
\definecolor{links}{HTML}{1a663b}
\definecolor{MyLightMagenta}{cmyk}{0.1,0.8,0,0.1}
\hypersetup{
colorlinks,citecolor=blue,filecolor=black,linkcolor=blue,urlcolor=blue
}
\begin{document}

\title{Unidirectional substitutes and complements\thanks{I am grateful to Ning Sun, Qianfeng Tang, and Guoqiang Tian for their continuous support and encouragement. All errors are mine.}}
\author{Chao Huang\thanks{Institute for Social and Economic Research, Nanjing Audit University. Email: huangchao916@163.com.}}
\date{}
\maketitle

\begin{abstract}
In discrete matching markets, substitutes and complements can be unidirectional between two groups of workers when members of one group are more important or competent than those of the other group for firms. We show that a stable matching exists and can be found by a two-stage Deferred Acceptance mechanism when firms' preferences satisfy a unidirectional substitutes and complements condition. This result applies to both firm-worker matching and controlled school choice. Under the framework of matching with continuous monetary transfers and quasi-linear utilities, we show that substitutes and complements are bidirectional for a pair of workers.
\end{abstract}



\section{Introduction}\label{Sec_Intro}

In a matching market, suppose workers are divided into two groups, where workers from one group are more important or competent than workers from the other group for firms. Then, firms' selections on the former group usually affect their selections on the latter group, but not vice versa. For instance, in many real-life sectors where firms match with skilled and unskilled workers, firms often let skilled workers help and guide unskilled workers. Thus, a skilled worker would possibly have an effect on some firms' decisions in hiring an unskilled worker. However, the availability of an unskilled worker probably does not affect a firm's selections on skilled workers. Consider the following preferences of firms $f_1, f_2$ over a skilled worker $s$ and an unskilled worker $u$.

\begin{equation}\label{exam_in1}
f_1: \{s,u\}\succ\{s\}\succ\emptyset \qquad\qquad\qquad f_2: \{s\}\succ\{u\}\succ\emptyset
\end{equation}

We notice that $s$ is a \textbf{complement} to $u$ for $f_1$ because $f_1$ does not want to hire $u$ when choosing from the set $\{u\}$ and would hire $u$ when the available set expands to $\{s,u\}$. However, $u$ exerts no complementary effect on $s$ for both firms, and thus $u$ is not a complement to $s$. Therefore, the complementarity in the above preference profile is unidirectional. We also notice that $s$ is a \textbf{substitute} to $u$ for $f_2$ because $f_2$ wants to hire $u$ when choosing from the set $\{u\}$ and would not hire $u$ when the available set expands to $\{s,u\}$. Nonetheless, $u$ exerts no substitutable effect on $s$ for both firms, and thus $u$ is not a (better) substitute to $s$. Therefore, the substitutability in the above preference profile is also unidirectional.\footnote{See Definition \ref{sc} for the formal definition of complements and substitutes in discrete matching markets.}

In another matching market where firms want to hire some managers and some assistants, the selections on managers are often more important than those on assistants for firms. The availability of an assistant probably does not affect a firm's decision on whether to hire a certain manager. However, the availability of a certain manager would possibly affect its selection on assistants. For instance, when a firm $f$ have the following preference in (\ref{exam_in2}) over two managers $m_1,m_2$, and two assistants $a_1,a_2$, substitutes and complements are also unidirectional: $m_1$ is a complement to $a_1$;\footnote{$m_1$ is a complement to $a_1$ since $a_1$ is not chosen by $f$ from $\{a_1\}$ but chosen by $f$ from $\{a_1\}\cup\{m_1\}$. For a similar reason, $m_2$ is a complement to $a_2$.} $m_2$ is a complement to $a_2$; $m_1$ is a substitute to $a_2$.\footnote{$m_1$ is a substitute to $a_2$ since $a_2$ is chosen by $f$ from $\{m_2,a_2\}$ while not chosen by $f$ from $\{m_2,a_2\}\cup\{m_1\}$.} There are no other substitute or complement between managers and assistants.

\begin{equation}\label{exam_in2}
f: \{m_1,a_1\}\succ\{m_1\}\succ\{m_2,a_2\}\succ\{m_2\}\succ\emptyset
\end{equation}

To my best knowledge, unidirectional substitutes or unidirectional complements have not been studied in matching theory. The substitutes condition (\citealp{RS90}), which is sufficient for the existence of a stable matching, rules out all complements, including the unidirectional ones. The complements condition rules out all substitutes, including the unidirectional substitutes.\footnote{The complements condition alone is not appealing in matching theory, since a stable matching may not exist when there are complementary preferences. The condition of \cite{O08} contains a cross-side complementarity assumption, which rules out all substitutes that cross the sides of a supply chain network, see Section \ref{Sec_USC}.} Previous theoretical and applied studies in the discrete matching market mostly rely on the substitutes condition. Positive results for the existence of a stable matching with complements are rather scarce.

This paper studies matching problems where workers are divided into two groups, where workers in one group are more important or competent than those in the other group for firms. We show that a stable matching always exists when firms' preferences satisfy a \textbf{unidirectional substitutes and complements} (henceforth USC) condition: For each firm, (i) there are no complements within the same group, and (ii) workers of one group are neither substitutes nor complements to workers of the other group. Assumption (ii) means that we only allow workers of the latter group to be unidirectional substitutes or unidirectional complements to workers of the former group. Both (\ref{exam_in1}) and (\ref{exam_in2}) are instances of firms' preferences that satisfy the USC condition.\footnote{Workers in (\ref{exam_in1}) are divided into two groups where one group contains a skilled worker, and the other contains an unskilled worker. Thus, firms' preferences trivially satisfy the assumption (i). It is also easy to check that there are no complements within the same group in (\ref{exam_in2}).} The USC condition subsumes the substitutes condition and allows for a new class of preferences that satisfy neither the same-side substitutability and cross-side complementarity (henceforth SSCC) condition of \cite{O08} nor the unimodularity condition of \cite{H21}, see Section \ref{Sec_USC}.\footnote{The USC condition reduces to the substitutes condition when one of the two groups is empty.}

Our proof of the existence theorem is constructive and intuitive. Under the USC condition, the Deferred Acceptance (henceforth DA) mechanism does not necessarily produce a stable matching due to the presence of complementarities. However, we can always find a stable matching by running a two-stage DA mechanism. We proceed with the problem of matching firms with skilled and unskilled workers. We let skilled workers propose to firms in the first stage and let unskilled workers propose to firms in the second stage. Because unskilled workers are not substitutes to skilled workers, the second stage does not change the allocation of skilled workers determined by the first stage. We then find that the complementarities that skilled workers exert on unskilled workers do not distort stability in the second stage.

We can also apply this result to controlled school choice with a particular type of proportionality constraints, which are very common in China. Cities in China are divided into school districts, where each district contains several high schools and primary schools. Schools are allowed to admit students of other districts but should prioritize students within their districts over cross-district students. Families should pay a sum of money called ``school selection fees'' to a school outside their districts once their children get permitted to attend the school. Admitting cross-district students is profitable for schools, but the government imposes ceilings on the proportionalities of cross-district students for schools. We show that schools' preferences satisfy the USC condition in this problem. Therefore, there always exists a stable matching that can be found by running the two-stage DA mechanism.

The two-sided many-to-one matching model in this paper can be extended to more general settings in different directions, see Section \ref{Sec_Concl}. We leave for future research on economic applications under those extended frameworks. In this paper, we would like to present our results in the basic setting of the two-sided many-to-one matching for the interest of brevity.

Discrete matching problems assume that there are no monetary transfers (school choice and college admission), or workers' wages are exogenously given (firm-worker matching and hospital-doctor matching). However, some studies in matching theory assume there are continuous monetary transfers between firms and workers where firms have quasi-linear utilities, and workers' wages are determined endogenously. An interesting phenomenon is that there are parallel existence results between the two problems: the gross substitutes condition of \cite{KC82} and the substitutes condition of \cite{RS90}, the gross substitutes and complements (henceforth, GSC) condition of \cite{SY06,SY09} and the SSCC condition of \cite{O08},\footnote{See also the full substitutability condition of \cite{HKNOW13}, which generalizes the GSC condition of \cite{SY06,SY09} into the framework of trading networks.} and the unimodularity conditions of \cite{DKM01} (recovered independently by \citealp{BK19}) and \cite{H21}. We find that the USC condition is an exception. For convenience, we call the matching market with continuous monetary transfers and quasi-linear utilities the quasi-linear market. Consider the quasi-linear market with a skilled worker $s$ and an unskilled worker $u$. Suppose a firm has the following valuation $v$ on the workers.

\begin{equation}\label{exam_in3}
v(\{s,u\})=x \qquad\quad v(\{s\})=5 \qquad\quad v(\{u\})=3 \qquad\quad v(\emptyset)=0
\end{equation}

The counterpart to the USC condition would require that the  salary for $u$ does not affect the firm's demand for $s$. However, this would require $x=8$,\footnote{For example,(i) if $x>8$, suppose $x=10$ and the firm has to pay 6 to $s$, then the firm would hire both workers when the salary for $u$ is 3, but would not hire any worker when the salary for $u$ is 5; (ii) if $x<8$, suppose $x=6$ and the firm has to pay 4 to $s$, then the firm would hire $s$ when the salary for $u$ is 3, but would hire $u$ when the salary for $u$ is 1.} with which the salary for $s$ does not affect the firm's demand for $u$ either.\footnote{When $x=8$, the firm's demand for $u$ only depends on $u$'s salary.} This is in contrast to (\ref{exam_in1}) where $s$ would affect the firms' demand for $u$, while $u$ does not affect the firms' demand for $s$.

To study the difference in the structures of substitutabilities and complementarities between the two markets, we introduce notions of substitutes and complements, which are natural counterparts to those in discrete matching, for a pair of workers in the quasi-linear market. We say that worker $w$ is a substitute (complement) to worker $w'$ for some firm if a lower salary for $w$ may decrease (resp. increase) the demand of $w'$ for this firm.\footnote{See Definition \ref{Def_scm} for the formal definition of substitutes and complements in the quasi-linear market.} We then find that the substitutes and complements must be bidirectional for a pair of workers in the quasi-linear market: If worker $w$ is a substitute (complement) to worker $w'$ for some firm, then $w'$ is also a substitute (resp. complement) to $w$ for this firm. This result shows that the structures of substitutabilities and complementarities are different between the discrete matching market and the quasi-linear market.

\subsection{Related literature}\label{lite}

The seminal work of \cite{GS62} proposed the problem of the two-sided matching and the DA mechanism in marriage markets and college admissions. \cite{KC82} studied the firm-worker matching in a quasi-linear market and found that a DA-like auction finds a stable matching when firms' valuations satisfy the gross substitutes condition. \cite{R84,R85} and \cite{RS90} developed the discrete matching into the job markets and found that the DA mechanism produces a stable matching when agents' preferences satisfy the substitutes condition.

The GSC condition of \cite{SY06,SY09} generalizes the gross substitutes condition in the quasi-linear market. The SSCC condition of \cite{O08} generalizes the substitutes condition in the discrete matching market and into a more general framework of supply chain networks. Our USC condition shares a similar structure with the GSC condition and the SSCC condition. Under the framework of the two-sided matching, the GSC condition and the SSCC condition both require that workers are divided into two groups where workers within the same group are not complements for firms, and workers of different groups are not substitutes for firms. The USC condition replaces the no cross-group substitutes assumption of the SSCC condition by the requirement that workers of one group do not affect firms' choices on workers of the other group. The GSC condition and the SSCC condition apply to the cases where workers of two groups are complementary for firms, while the USC condition applies to the cases where workers of one group are more important or competent than those of the other group for firms. We can extend our model to a structure where workers are divided into more than two groups, and further to many-to-many matching or network structure as that of \cite{O08}, see Section \ref{Sec_Concl}. A recent work of the author (\citealp{H21}) proposed a unimodularity condition, which is independent of the substitutes condition. The USC condition allows firms' preferences to satisfy neither the SSCC condition nor the unimodularity condition.

Other than firm-worker matching, we have also applied our result to controlled school choice. In school-choice programs, schools often concern affirmative action, ethnic balance, or socioeconomic factors on the distribution of students. Complementarities arise in the program when there are type-specific floors, diversity considerations, or proportionality constraints. \cite{EHYY14} and \cite{NV19} respectively treated type-specific floors and proportionality constraints as soft to preserve stability. The proportionality constraint studied in our paper is a special case of \cite{NV19}, where schools' preferences satisfy the USC condition. \cite{EY15} treated diversity considerations by designing schools' choice rules to both fit the substitutes condition and achieve diversity. \cite{HKY19} studied an interdistrict school-choice program which allocates students by a mechanism over districts, based on the current interdistrict program in America. We consider the placements of within-district and cross-district students in a single district, because the school-choice programs of different districts are independent with each other in China.

Complementary preferences have been an important issue in two-sided matching during the past decade. Thus, this paper also relates to an extensive ongoing literature on complementarities in matching markers, which includes matching with couples (see e.g., \citealp{KK05}, \citealp{KPR13}, \citealp{ABH14}, and \citealp{NV18}), matching with peer effects (see e.g., \citealp{EY07}, \citealp{P12}), and matching in large markets (see e.g., \citealp{AWW13}, \citealp{AH18}, and \citealp{CKK19}), among others.

The remainder of this paper is organized as follows. Section \ref{Sec_M} introduces the discrete matching market and the USC condition. Section \ref{Sec_Mech} presents the two-stage DA mechanism and our result. Section \ref{Sec_App} applies our result to controlled school choice. Section \ref{Sec_Money} shows that substitutes and complements are bidirectional for a pair of workers in the quasi-linear market. Section \ref{Sec_Concl} concludes. Proofs are relegated into the Appendix.

\section{Model\label{Sec_M}}

\subsection{Discrete matching market \label{Sec_MM}}

There is a set $F=\{f_1,\ldots,f_m\}$ of $m$ firms and a set $W$ of $n$ workers. Let ${\o}$ be the null firm, representing not being matched with any firm. Each worker $w\in W$ has a strict, transitive and complete preference $P_w$ over $\widetilde{F}:=F\cup\{{\o}\}$. For any $f, f'\in \widetilde{F}$, we write $f\succ_w f'$ when $w$ prefers $f$ to $f'$ according to $P_w$. We write $f\succeq_w f'$ if either $f\succ_w f'$ or $f=f'$. Let $P_W$ denote the preference profile of all workers. Each firm $f\in F$ has a strict, transitive and complete preference $P_f$ over $2^W$. For any $X,X'\subseteq W$, we write write $X\succ_f X'$ when $f$ prefers $X$ to $X'$ according to $P_f$. We write $X\succeq_f X'$ if either $X\succ_f X'$ or $X=X'$. Let $P_F$ be the preference profile of all firms. A matching market can be summarized as a tuple $\Gamma=(F,W,P_F,P_W)$.

\begin{definition}\label{Def_matching}
\normalfont
A \textbf{matching} $\eta$ is a function from the set $\widetilde{F}\cup W$ into $\widetilde{F}\cup 2^W$ such that for all $f\in \widetilde{F}$ and $w\in W$,
\begin{description}
\item[1.] $\eta(w)\in \widetilde{F}$;

\item[2.] $\eta(f)\in 2^W$;

\item[3.] $\eta(w)=f$ if and only if $w\in\eta(f)$.
\end{description}
\end{definition}

Let $Ch_f$ be the choice function of $f$ such that for any $X\subseteq W$, $Ch_f(X)\subseteq X$ and $Ch_f(X)\succeq_f X'$ for any $X'\subseteq X$. By convention, let $Ch_{{\o}}(X)=X$ for all $X\subseteq W$. We say a matching $\eta$ is \textbf{individual rational} if $\eta(w)\succeq_w {\o}$ for all $w\in W$ and $\eta(f)=Ch_f(\eta(f))$ for all $f\in F$. We say a firm $f$ and a subset of workers $X\subseteq W$ form a \textbf{blocking coalition} that blocks $\eta$ if $f\succeq_w\eta(w)$ for all $w\in X$, and $X\succ_f\eta(f)$. In words, individual rationality requires that each matched worker prefer her current employer to being unmatched and that no firm wish to unilaterally drop any of its employees. $f$ and $X$ blocks $\eta$ if $f$ strictly prefers $X$ to its current set of employees $\eta(f)$, and each worker $w$ in $X$ weakly prefers $f$ to her current employer $\eta(w)$.

\begin{definition}\label{stability}
\normalfont
A matching $\eta$ is \textbf{stable} if it is individual rational and there is no blocking coalition that blocks $\eta$.\footnote{The no blocking coalition condition implies the individual rationalities of firms: $\eta(f)=Ch_f(\eta(f))$ for all $f\in F$. The exposition of this definition follows the convention in the literature.}
\end{definition}

We now introduce the following concepts of substitutes and complements.

\begin{definition}\label{sc}
\normalfont
$w$ is a \textbf{substitute} to $w'$ for firm $f$, if there exists $X\subseteq W$ such that $\{w,w'\}\subseteq X$, $w'\in Ch_f(X\setminus \{w\})$, and $w'\notin Ch_f(X)$. $w$ is a \textbf{complement} to $w'$ for firm $f$, if there exists $X\subseteq W$ such that $\{w,w'\}\subseteq X$, $w'\notin Ch_f(X\setminus \{w\})$, and $w'\in Ch_f(X)$.
\end{definition}

In words, $w$ is a substitute to $w'$ if the firm wants to drop $w'$ when $w$ becomes available in some case. $w$ is a complement to $w'$ if the firm switches to employ $w'$ when $w$ becomes available in some case. It is possible that a worker is both a substitute and a complement to another worker for the same firm. For instance, consider the following preference for firm $f$ on three workers $w_1, w_2,$ and $w_3$.
\begin{equation*}
\{w_1,w_3\}\succ \{w_2,w_3\}\succ \{w_1,w_2\}\succ \{w_1\}\succ \emptyset
\end{equation*}
$w_1$ is a substitute to $w_2$ for the firm because $w_2\in Ch_f(\{w_2,w_3\})$ and $w_2\notin Ch_f(\{w_2,w_3\}\cup\{w_1\})$. $w_1$ is also a complement to $w_2$ for the firm because $w_2\notin Ch_f(\{w_2\})$ and $w_2\in Ch_f(\{w_2\}\cup\{w_1\})$.

There always exists a stable matching when firms' preferences satisfy the substitutes condition of \cite{RS90}, the definition of which can be given with our new terminology. Firms' preferences satisfy the substitutes condition if there are no complements.\footnote{The original definition for substitutes condition is as follows: Firm $f$'s preference satisfy the substitutes condition if for any $X\subseteq W$ and $w,w'\in W$ with $\{w,w'\}\subseteq X$, $w'\in Ch_f(X)$ implies $w'\in Ch_f(X\setminus\{w\})$. This definition is equivalent to Definition \ref{Def_sub}.\label{foot}}

\begin{definition}\label{Def_sub}
\normalfont
Firm $f$'s preference satisfies the \textbf{substitutes} condition if for any $w,w'\in W$, $w$ is not a complement to $w'$ for $f$.
\end{definition}

\subsection{Unidirectional substitutes and complements\label{Sec_USC}}

Suppose the worker set $W$ is partitioned into two subsets $S$ and $U$. We present our model in the context of matching firms with skilled and unskilled workers, and thus we call elements from $S$ skilled workers and elements from $U$ unskilled workers. The pattern of division is shared by all firms, and thus a worker considered as unskilled by one firm cannot be considered as skilled by another firm. Now we introduce the USC condition.

\begin{definition}\label{Def_usc}
\normalfont
A firm's preference satisfies the \textbf{unidirectional substitutes and complements} condition if for this firm,

\begin{description}
\item[(i)]for any $s,s'\in S$, $s$ is not a complement to $s'$; for any $u,u'\in U$, $u$ is not a complement to $u'$;

\item[(ii)]for any $s\in S$ and $u\in U$, $u$ is neither a substitute nor a complement to $s$.
\end{description}
\end{definition}

The assumption (i) of the USC condition says that we assume no complements within either group. The assumption (ii) of the USC condition means that we require the substitutes and complements to be unidirectional between the two groups. To see the role of the assumption (ii), we provide two simple examples below, both of which violate this assumption. Recall firms' preferences of (\ref{exam_in1}), which satisfy the USC condition. It is not difficult to check that a stable matching exists in this market for all possible preferences of workers. We modify (\ref{exam_in1}) into the following two preference profiles.

\begin{equation}\label{exam_usc1}
\begin{aligned}
&f_1: \{s,u\}\succ\emptyset \qquad\qquad\qquad\qquad &s: &\quad f_1\succ f_2\\
&f_2: \{s\}\succ\{u\}\succ\emptyset \qquad\qquad\qquad\qquad &u: &\quad f_2\succ f_1
\end{aligned}
\end{equation}

$f_1$'s preference in (\ref{exam_usc1}) violates the assumption (ii) of the USC condition, since $s$ and $u$ are mutual complements for $f_1$. No stable matching exists in this market.\footnote{In this market, $f_1$ should hire both workers or neither of them in any stable matching. In the former case, $f_2$ would form a blocking coalition with $u$, who prefers $f_2$ to $f_1$. In the latter case, $f_2$ would be matched with $s$, leaving $u$ unmatched. Then $f_1$ would form a blocking coalition with both $s$ and $u$.} Consider another market as follows.

\begin{equation}\label{exam_usc2}
\begin{aligned}
&f_1: \{s,u\}\succ\{s\}\succ\emptyset \qquad\qquad\qquad\qquad &s: &\quad f_2\succ f_1\\
&f_2: \{u\}\succ\{s\}\succ\emptyset \qquad\qquad\qquad\qquad &u: &\quad f_1\succ f_2
\end{aligned}
\end{equation}

Firms' preferences in (\ref{exam_usc2}) violate the assumption (ii) of  the USC condition, since $s$ is a complement to $u$ for $f_1$, while $u$ is a substitute to $s$ for $f_2$. No stable matching exists in this market either.\footnote{If $f_1$ employs both $s$ and $u$, then $f_2$ would form a blocking coalition with $s$. If $f_1$ employs $s$, and $f_2$ employs $u$, then $f_1$ would form a blocking coalition with both $s$ and $u$. If $f_1$ is matched with the empty set, $f_2$ would employ $u$, leaving $s$ unmatched. Then $f_1$ would also form a blocking coalition with both $s$ and $u$.}

The structure of the USC condition resembles that of the SSCC condition. \cite{O08} defined the SSCC condition in a supply chain network. We compare the two conditions in the two-sided matching market, where the SSCC condition is described as follows.

\begin{definition}
\normalfont
A firm's preference satisfies the \textbf{same-side substitutability and cross-side complementarity} condition if for this firm,

\begin{description}
\item[(i)]Same-side substitutability: for any $s,s'\in S$, $s$ is not a complement to $s'$; for any $u,u'\in U$, $u$ is not a complement to $u'$;

\item[(ii)]Cross-side complementarity: for any $s\in S$ and $u\in U$, $u$ is not a substitute to $s$, $s$ is not a substitute to $u$.
\end{description}
\end{definition}

The assumption (i) of the USC condition is the same as the assumption (i) of the SSCC condition. The two conditions differ in the assumption (ii). The SSCC condition rules out all substitutes between the two groups, while the USC condition rules out all substitutabilities and complementarities that unskilled workers exert on skilled workers. Both firms' preference profiles of (\ref{exam_in1}) and (\ref{exam_in2}) satisfy the USC condition but violate the assumption (ii) of the SSCC condition: $s$ is a substitute to $u$ for $f_2$ in (\ref{exam_in1}), and $m_1$ is a substitute to $a_2$ for $f$ in (\ref{exam_in2}).

We then introduce the unimodularity condition informally. The demand type $\mathcal{D}_f$ for firm $f$ consists of indicators for changes of $f$'s choices as its available set of workers expands. In the example of (\ref{exam_in1}), $f_1$ is of demand type $\{(0,1),(1,0),(1,1)\}$. $(0,1)$ is in $f_1$'s demand type because when the available set for $f_1$ expands from $\{s\}$ to $\{s,u\}$, $f_1$'s choice changes from $\{s\}$ to $\{s,u\}$, then $(0,1)$ is the indicator for $\{s,u\}$ minus $\{s\}$. $(1,0)$ and $(1,1)$ are due to the cases when $f$'s available set expands from $\emptyset$ to $\{s\}$ and from $\emptyset$ to $\{s,u\}$, respectively. Likewise, $f_2$ is of demand type $\{(0,1),(1,0),(1,-1)\}$. Note that $(1,-1)$ is in $f_2$'s demand type because when the available set for $f_2$ expands from $\{u\}$ to $\{s,u\}$, $f_2$'s choice changes from $\{u\}$ to $\{s\}$, then $(1,-1)$ is the indicator for $\{s\}$ minus $\{u\}$. The demand type of firms' preference profile is the union of all firms' demand types: firms' preference profile in (\ref{exam_in1}) is of demand type $\{(0,1),(1,0),(1,1),(1,-1)\}$, which is not unimodular since the determinant of $(1,1)$ and $(1,-1)$ is $-2$. Therefore, (\ref{exam_in1}) is an instance of firms' preferences that satisfy the USC condition but violate both the unimodularity condition and the SSCC condition. We refer the reader to \cite{H21} for the formal definition and more details for the unimodularity condition.

\section{Mechanism\label{Sec_Mech}}

When firms' preferences satisfy the substitutes condition, we can obtain a stable matching by running the worker-proposing DA mechanism, which operates as follows:
\medskip
\begin{description}
\item[Step $1$] Each worker proposes to her most preferred firm. Each firm tentatively accepts its most preferred set of workers from the proposals it receives and rejects the rest.

\item[Step $k,k\geq2$] Each worker rejected in the previous round proposes to her next best firm. Each firm that faces new applicants tentatively accepts its most preferred set of workers and rejects the rest, among both new applicants and previously accepted workers.

The algorithm stops when there are no rejections.
\end{description}
\medskip

Under the USC condition, the worker-proposing DA mechanism does not necessarily produce a stable outcome, see Example \ref{exam_illu} below. However, we can always find a stable matching by running a two-stage DA mechanism. In the \textbf{two-stage worker-proposing DA mechanism}, skilled workers propose to firms in the first stage, and unskilled workers propose to firms in the second stage. Formally, in the first stage,

\medskip
\begin{description}
\item[Step $1.1$] Each skilled worker proposes to her most preferred firm. Each firm tentatively accepts its most preferred set of skilled workers from the proposals it receives and rejects the rest.

\item[Step $1.k,k\geq2$] Each skilled worker rejected in the previous round proposes to her next best firm. Each firm that faces new applicants tentatively accepts its most preferred set of skilled workers and rejects the rest, among both new applicants and previously accepted workers.

The first stage terminates when there are no rejections.
\end{description}
\medskip
In the second stage,
\medskip
\begin{description}
\item[Step $2.1$] Each unskilled worker proposes to her most preferred firm. Each firm that faces new applicants tentatively accepts its most preferred set of workers and rejects the rest, among both new applicants and the skilled workers allocated to the firm by the first stage.

\item[Step $2.k,k\geq2$] Each worker rejected in the previous round proposes to her next best firm. Each firm that faces new applicants tentatively accepts its most preferred set of workers and rejects the rest, among both new applicants and previously accepted workers.

The second stage terminates when there are no rejections, and the firms finalize their matches to the tentatively accepted workers.
\end{description}

Because unskilled workers are not substitutes to skilled workers, in each step of the second stage only unskilled workers are rejected. The second stage does not change the allocation of skilled workers determined by the first stage. We then know that whenever an unskilled worker is rejected by a firm at any step of the second stage, this unskilled worker would not be reconsidered as desirable for the firm after this step. Therefore, although skilled workers may be complements to unskilled workers, these complementarities do not distort stability during the two stages.

\begin{theorem}\label{thm_main}
\normalfont
When each firm's preference satisfies the USC condition, the two-stage worker-proposing DA mechanism produces a stable matching.
\end{theorem}

\begin{example}\label{exam_illu}
\normalfont
Consider a matching market $\Gamma$ where there are three firms $f_1,f_2,f_3$, three skilled worker $s_1,s_2,s_3$, and two unskilled worker $u_1,u_2$. The preferences of firms and workers are as follows.
\begin{align*}
&f_1: \{s_2,u_2\}\succ\{s_2\}\succ\{s_1\}\succ\{s_3\}\succ\emptyset \quad &s_1: &\quad f_1\succ f_3\succ f_2\succ{\o}\\
&f_2: \{s_1\}\succ\{s_3,u_1\}\succ\{s_3\}\succ\emptyset \quad &s_2: &\quad f_3\succ f_1\succ{\o}\\
&f_3: \{s_1\}\succ\{s_2\}\succ\emptyset \quad &s_3: &\quad f_1\succ f_2\succ{\o}\\
&\quad &u_1: &\quad f_2\succ{\o}\\
&\quad &u_2: &\quad f_1\succ{\o}
\end{align*}

The firms' preferences above do not satisfy the SSCC condition but satisfy the USC condition.\footnote{The SSCC condition does not hold because $f_2$'s preference does not satisfy the cross-side complementarity assumption: $u_1\in Ch_{f_2}(\{s_3,u_1\})$ but $u_1\notin Ch_{f_2}(\{s_3,u_1\}\cup\{s_1\})$.} The worker-proposing DA mechanism operates as follows.

\bigskip
\begin{center}
\begin{tabular}
{cccc}%
$f_1$ & $f_2$ & $f_3$ & ${\o}$\\\hline
\fbox{$s_1$} $s_3$ $u_2$ & $u_1$ & \fbox{$s_2$}&\\
& $\fbox{$s_3$}$ & & \fbox{$u_1$} \fbox{$u_2$} \\\hline
\fbox{$s_1$}& $\fbox{$s_3$}$ & \fbox{$s_2$} & \fbox{$u_1$} \fbox{$u_2$}%
\end{tabular}
\end{center}
\bigskip

The matching produced above is not stable because $f_2$ and $\{s_3,u_1\}$ form a blocking coalition. The two-stage worker-proposing DA mechanism produces a stable matching $\eta_1$ as follows.\\

\begin{center}
Stage 1:\quad
\begin{tabular}
{cccc}%
$f_1$ & $f_2$ & $f_3$ & ${\o}$ \\\hline
\fbox{$s_1$} $s_3$ & &\fbox{$s_2$} & \\
& \fbox{$s_3$} & & \\\hline
\fbox{$s_1$} & \fbox{$s_3$} &\fbox{$s_2$} &
\end{tabular}
\qquad Stage 2:\quad
\begin{tabular}
{cccc}%
$f_1$ & $f_2$ & $f_3$ & ${\o}$ \\\hline
\fbox{$s_1$} $u_2$ & \fbox{$s_3$} \fbox{$u_1$} &\fbox{$s_2$} & \\
& & & \fbox{$u_2$} \\\hline
\fbox{$s_1$} & \fbox{$s_3$} \fbox{$u_1$} &\fbox{$s_2$} & \fbox{$u_2$}
\end{tabular}
\end{center}

\begin{equation*}
\eta_1=\left(
             \begin{aligned}
             f_1  \quad & \quad f_2 \quad & \quad f_3 \quad & \quad {\o} \\
             s_1 \quad&\quad s_3,u_1 \quad&\quad s_2 \quad&\quad u_2
             \end{aligned}
\right)
\end{equation*}

There are two stable matchings in this market, the other one is:

\begin{equation*}
\eta_2=\left(
             \begin{aligned}
             f_1  \quad & \quad f_2 \quad & \quad f_3 \quad & \quad {\o} \\
             s_2,u_2 \quad&\quad s_3,u_1 \quad&\quad s_1 \quad&\quad
             \end{aligned}
\right)
\end{equation*}

In marriage markets, college admissions, and firm-worker matching under substitutes condition, there are polarizations of interests between agents of opposite sides on the stable outcomes: Among all stable matchings, there exists a stable matching that is best for firms (men, colleges) and worst for workers (resp. women, students), as well as a stable matching worst for firms (men, colleges) and best for workers (resp. women, students) (see \citealp{GS62}, \citealp{R84}, and \citealp{RS90}). This polarization does not hold under the USC condition: In this example, $u_2$, $f_1$, and $f_3$ are better at $\eta_2$ than at $\eta_1$, while $s_1$ and $s_2$ are better at $\eta_1$ than at $\eta_2$.
\end{example}

When firms' preferences satisfy the substitutes condition, we can also find a stable matching by running the firm-proposing DA mechanism. We refer the reader to Chapter 6 of \cite{RS90} on the details of this mechanism. Under the USC condition, we can also apply the firm-proposing DA to the two-stage DA mechanism in the first, the second, or both stages. We can obtain a stable matching by each of the following three types of two-stage DA mechanisms: (i) firms propose to skilled workers in the first stage, and unskilled workers propose to firms in the second stage; (ii) skilled workers propose to firms in the first stage, and firms propose to unskilled workers in the second stage; (iii) firms propose to skilled workers in the first stage, and firms propose to unskilled workers in the second stage.\footnote{When we let firms propose to unskilled workers in the second stage, we should adjust the procedure accordingly. For instance, in the first step of the second stage each firm $f\in F$ should propose to its most preferred set of workers that includes all of those skilled workers allocated to $f$ by the first stage, but does not include any other skilled worker.} We omit the details and proofs for these variant mechanisms.

\section{Application\label{Sec_App}}

This section applies our result to controlled school choice with a particular type of proportionality constraints, which are very common in China. Cities in China are often divided into several school districts, each of which contains several middle schools and primary schools. Families are supposed to send their children to schools inside the family's district. However, some families are not satisfied with schools inside their districts and want to send their children to better schools outside their districts. The government allows schools to admit cross-district students but requires each school to give priority to students inside its district over those outside its district. Priorities within the same group (within-district students or cross-district students) are determined by a combination of geographic distances, school performances, and other social factors. Families outside the district of a school should pay  ``school selection fees'' to the school once their children get permitted to attend the school. Thus it is profitable for schools to admit cross-district students. However, local governments often impose ceilings on the proportionalities of cross-district students for schools. In this problem, we show that schools' preferences satisfy the USC condition: Although students in the school's district can be substitutes or complements to cross-district students, the latter are neither substitutes nor complements to the former. Therefore, there always exists a stable matching that can be found by running the two-stage DA mechanism.

For convenience, we elaborate with the notations of Section \ref{Sec_M}. Now $F$ is a set of schools that belongs to a school district; ${\o}$ is the null school, representing the outside options for students. $W$ is a set of students. Each school $f\in F$ has capacity $q_f$ and a strict priority order $\rhd_f$ on $W$. We write $w\rhd_f w'$ to indicate that $w$ has higher priority than $w'$ at $f$. $W$ is partitioned into $S$ and $U$ where $S$ is the set of students who live in the district, and $U$ the set of students from other districts. Each within-district student has higher priority than any of the cross-district students at each school: For each $s\in S$, $u\in U$, and $f\in F$, we have $s\rhd_f u$. For any $w\in W$, $f\in F$ and $X\subseteq W$, let $X^{\rhd_f w}\equiv\{w'\in X\mid w'\rhd_f w\}$ be the subset of $X$ that contains students who have higher priority than $w$ at $f$.

The government assigns each school $f\in F$ a proportionality ceiling $\alpha_f\in[0,1]$ on the cross-district students admitted by the school. Admitting cross-district students is prohibited for $f$ if $\alpha_f=0$ and free for $f$ if $\alpha_f=1$. Schools' preferences are described indirectly by their choice functions. For any $X\subseteq W$, the choice function $Ch_f(X)$ for each $f\in F$ is defined as follows: $Ch_f(X)\subseteq X$ and

\begin{align}
&\text{for each}\quad s\in X\cap S,\quad s\in Ch_f(X)\quad \text{if}\quad \mid X^{\rhd_f s} \mid<q_f;\label{sc1}\\
&\text{for each}\quad u\in X\cap U,\quad u\in Ch_f(X)\quad \text{if}\quad \mid X^{\rhd_f u} \mid<q_f \quad\text{and}\quad 1-\frac{\mid X\cap S\mid}{\mid X^{\rhd_f u} \mid+1}\leq\alpha_f.\label{sc2}
\end{align}

\begin{theorem}\label{thm_school}
\normalfont
In the problem of the controlled school choice with proportionality ceilings on cross-district students, schools' preferences satisfy the USC condition.
\end{theorem}

\begin{example}
\normalfont
Suppose there are two schools $f_1,f_2$, two within-district students $s_1,s_2$, and two cross-district students $u_1,u_2$. The priorities for schools and the preferences for students are as follows.
\begin{align*}
&f_1:\quad s_1 \rhd s_2 \rhd u_1 \rhd u_2\quad \quad &s_1: &\quad f_2\succ f_1\succ {\o}\\
&f_2:\quad s_2 \rhd s_1 \rhd u_2 \rhd u_1\quad \quad &s_2: &\quad f_2\succ f_1\succ{\o}\\
&&u_1: &\quad f_1\succ f_2\succ{\o}\\
&&u_2: &\quad f_2\succ f_1\succ{\o}
\end{align*}
School $f_1$ has capacity $q_{f_1}=2$, and school $f_2$ has capacity $q_{f_1}=1$. The government imposes proportionality ceilings $\alpha_{f_1}=\frac{1}{2}$ and $\alpha_{f_2}=1$: The cross-district students allocated at $f_1$ should be no more than half of the number of students allocated at $f_1$, while there is no constraint for $f_2$. The student-proposing DA mechanism operates as follows.\\
\bigskip
\begin{center}
\begin{tabular}
{ccc}%
$f_1$ & $f_2$ & ${\o}$\\\hline
$u_1$ & $s_1$ \fbox{$s_2$} $u_2$\\
\fbox{$s_1$} \fbox{$u_2$}& \fbox{$s_2$} $u_1$ & \\
& & \fbox{$u_1$} \\\hline
\fbox{$s_1$} \fbox{$u_2$} & \fbox{$s_2$} & \fbox{$u_1$}%
\end{tabular}
\end{center}
\bigskip
In the first step, $u_1$ is rejected by $f_1$ due to the proportionality constraint. The matching produced above is not stable because $f_1$ and $\{s_1,u_1\}$ form a blocking coalition.
The two-stage student-proposing DA mechanism operates as follows. It produces a stable matching which allocates $\{s_1,u_1\}$ at $f_1$, $\{s_2\}$ at $f_2$, and $\{u_2\}$ at the null school ${\o}$.
\bigskip
\begin{center}
Stage 1:\quad
\begin{tabular}
{ccc}%
$f_1$ & $f_2$ & ${\o}$\\\hline
& $s_1$ \fbox{$s_2$}&\\
$\fbox{$s_1$}$& &\\\hline
\fbox{$s_1$} & \fbox{$s_2$}& %
\end{tabular}
\qquad Stage 2:\quad
\begin{tabular}
{ccc}%
$f_1$ & $f_2$ & ${\o}$\\\hline
$\fbox{$s_1$}$ \fbox{$u_1$}&\fbox{$s_2$} $u_2$ & \\
$\fbox{$s_1$}$ \fbox{$u_1$} $u_2$&  & \\
& & \fbox{$u_2$} \\\hline
$\fbox{$s_1$}$ \fbox{$u_1$}&\fbox{$s_2$}&\fbox{$u_2$} %
\end{tabular}
\end{center}
\bigskip
\end{example}

Current placements of cross-district students in China are in a decentralized manner that the families contact schools outside their districts on their own. A family may contact and get permitted by several schools of a district but finally choose only one of them. Our result suggests that the placements of cross-district students can be centralized with those of the within-district students in a two-stage DA mechanism, which could avoid such situations.

\section{Quasi-linear market\label{Sec_Money}}

In this section, we introduce notions of substitutes and complements, analogous to those in the discrete matching market, to the quasi-linear market of \cite{KC82}. In the quasi-linear market, firms have quasi-linear utilities, and there are continuous monetary transfers between firms and workers. We show that substitutes and complements are bidirectional for a pair of workers under this framework.

Let $W=\{w_1,\ldots,w_n\}$ be a set of $n$ workers. A firm has a valuation function $v: 2^W\rightarrow \mathbb{R}$, which specifies the value it gets from each possible subset of workers. Let $\textbf{p}\in \mathbb{R}^n$ be the salary vector where $p_i$ is the salary that the firm has to pay worker $w_i$ when the firm hires $w_i$.\footnote{We allow workers' salaries to be negative, and thus our model is equivalent to the setting where an agent has a quasi-linear utility on indivisible ``goods'' and ``bads''. The result in this section also holds if we assume the salaries of workers to be positive (i.e., $\textbf{p}\in \mathbb{R}_{++}^n$).} Given a salary vector $\textbf{p}\in \mathbb{R}^n$ and a subset of workers $X\subseteq W$, we also write as $p(w)$ the salary for any worker $w\in W$, and $c(X,\mathbf{p})=\sum_{w\in X}p(w)$ the cost of salaries for worker set $X$. The firm has to solve the following decision problem:
\begin{equation}\label{max}
\max_{X\subseteq W}\{v(X)-c(X,\mathbf{p})\}
\end{equation}
The set of solutions $D(\textbf{p})\equiv\arg\max_{X\subseteq W}\{v(X)-c(X,\mathbf{p})\}$ is called the demand correspondence. The setting of this section is equivalent to the one where an agent has a quasi-linear utility on indivisible objects, where the agent and the objects correspond to the firm and the workers, respectively.

Because firms may have multiple optimal choices in the quasi-linear market, there are no exact counterparts to the notions of substitutes and complements of the discrete matching market. We consider the following concepts of substitutes and complements for a pair of workers as natural counterparts to those in Definition \ref{sc}. We say that worker $w$ is \textbf{demanded} by the firm at salary vector $\mathbf{p}$ if there exists $X\subseteq W$ such that $w\in X$ and $X\in D(\mathbf{p})$.

\begin{definition}\label{Def_scm}
\normalfont
Worker $w$ is a \textbf{substitute} to worker $w'$ for the firm if there exists $\mathbf{p},\mathbf{p}'\in \mathbb{R}^n$ with $p(w)>p'(w)$ and $p(\widetilde{w})=p'(\widetilde{w})$ for all $\widetilde{w}\neq w$ such that $w'$ is demanded at $\mathbf{p}$ but not demanded at $\mathbf{p}'$. Worker $w$ is a \textbf{complement} to worker $w'$ for the firm if there exists $\mathbf{p},\mathbf{p}'\in \mathbb{R}^n$ with $p(w)>p'(w)$ and $p(\widetilde{w})=p'(\widetilde{w})$ for all $\widetilde{w}\neq w$ such that $w'$ is not demanded at $\mathbf{p}$ but demanded at $\mathbf{p}'$.
\end{definition}

In words, $w$ is a substitute to $w'$ if the firm wants to drop $w'$ when it costs a lower salary to hire $w$ in some case. $w$ is a complement to $w'$ if the firm switches to employ $w'$ when it costs a lower salary to hire $w$ in some case. We find that substitutes and complements for a pair of workers must be bidirectional in the quasi-linear market.

\begin{theorem}\label{thm_bidirection}
\normalfont
If $w$ is a substitute to $w'$ for the firm, then $w'$ is also a substitute to $w$ for the firm. If $w$ is a complement to $w'$ for the firm, then $w'$ is also a complement to $w$ for the firm.
\end{theorem}

We illustrate this result with the following example, which is in contrast to (\ref{exam_in1}).

\begin{example}
\normalfont
Firm $f_1$ and $f_2$ have valuations $v_1$ and $v_2$, respectively, on a skilled worker $s$ and an unskilled worker $u$:
\begin{eqnarray*}
v_1(\{s,u\})=8,\quad v_1(\{s\})=5,\quad v_1(\{u\})=-1,\quad v_1(\emptyset)=0\\
v_2(\{s,u\})=-1,\quad v_2(\{s\})=3,\quad v_2(\{u\})=2,\quad v_2(\emptyset)=0
\end{eqnarray*}

Let $D_1$ and $D_2$ be the demand correspondences of $f_1$ and $f_2$, respectively. When the salary for $s$ is 7 and the salary for $u$ is 2, $f_1$'s best choice is to hire nobody: $D_1((7,2))=\{\emptyset\}$. When the salary for $s$ decreases to 5, hiring both $s$ and $u$ is $f_1$'s best choice: $D_1((5,2))=\{\{s,u\}\}$. $u$ is not demanded at salary vector $(7,2)$, but demanded when $s$ only costs 5 for $f_1$. Hence, $s$ is a complement to $u$ for $f_1$. For $f_2$, we have $D_2((3,1))=\{\{u\}\}$; and $D_2((1,1))=\{\{s\}\}$. $u$ is demanded at salary vector $(3,1)$, but not demanded when $s$ only costs 1 for $f_2$. Hence, $s$ is a substitute to $u$ for $f_2$.

We then show that the complements and substitutes are bidirectional:

(1) Given the salary for $u$ is 2, because $\{s,u\}$ is better than $\emptyset$ for $f_1$ when $s$ costs $f_1$ the lower salary 5 and worse than $\emptyset$ for $f_1$ when $s$ costs $f_1$ the higher salary 7, we know that $f_1$ is indifferent between $\{s,u\}$ and $\emptyset$ when $s$ requires a certain salary, which is 6. We have $D_1((6,2))=\{\{s,u\},\emptyset\}$, then for all $\delta>0$, $D_1((6,2+\delta))=\{\emptyset\}$. Therefore, $u$ is a complement to $s$ for $f_1$ since $s$ is not demanded at $(6,2+\delta)$, but demanded at $(6,2)$.

Suppose instead $v(\{s\})=6$, then $D_1((7,2))=\{\emptyset\}$ and $D_1((5,2))=\{\{s,u\},\{s\}\}$ indicate that $s$ is still a complement to $u$ for $f_1$. We then have $D_1((6,2))=\{\{s,u\},\{s\},\emptyset\}$. In this case, there exists $\delta>0$ such that $D_1((6+\delta,2-\delta))=\{\{s,u\},\emptyset\}$, and then $D_1((6+\delta,2))=\{\emptyset\}$.\footnote{$D_1((6+\delta,2-\delta))=\{\{s,u\},\emptyset\}$ holds when $\delta<p(u)-v(\{u\})=2+1=3$, which guarantees that $\{u\}\notin D_1((6+\delta,2-\delta))$.} Therefore, $s$ is not demanded at $(6+\delta,2)$, but demanded at $(6+\delta,2-\delta)$.

(2) Given the salary for $u$ is 1, because $\{u\}$ is better than $\{s\}$ for $f_2$ when $s$ costs $f_2$ the higher salary 3 and worse than $\{s\}$ for $f_2$ when $s$ costs $f_2$ the lower salary 1, we know that $f_2$ is indifferent between $\{u\}$ and $\{s\}$ when $s$ requires a certain salary, which is 2. We have $D_1((2,1))=\{\{s\},\{u\}\}$, then for all $\delta>0$, $D_2((2,1-\delta))=\{\{u\}\}$. Therefore, $u$ is a substitute to $s$ for $f_2$ since $s$ is demanded at $(2,1)$, but not demanded at $(2,1-\delta)$.
\end{example}

\section{Concluding remarks\label{Sec_Concl}}

We can extend the discrete matching market in this paper in different directions. First, we can generalize the assumption of the two groups of workers to $N>2$ groups,\footnote{Under the framework of the two-sided matching, both \cite{SY06,SY09} and \cite{O08} require that workers are divided into exactly two groups. Their results do not extend to the cases where workers are divided into more than two groups.} where there are no within-group complements, and members of the $j$-th group are neither substitutes nor complements to members of the $k$-th group for all $j,k$ with $1\leq k<j\leq N$. One may find such structures when workers are divided into several levels, where workers of a higher level are more important or competent for firms than those of a lower level. We can find a stable matching in this case by running an $N$-stage DA mechanism.

We can further extend the many-to-one matching to many-to-many matching, where each worker works for a set of firms.\footnote{For many-to-many matching, see \cite{R84,R85}, \cite{S99}, and \cite{EO06}, among others.} When preferences of both skilled and unskilled workers satisfy the substitutes condition, and firms' preferences satisfy the USC condition, we can hopefully produce a desirable outcome by running a mechanism similar to our two-stage DA. We can also consider contract terms between firms and workers, where contracts may specify wages, insurances, retirement plans, and other job descriptions.\footnote{See e.g., \cite{R84} and \cite{HM05}.} Finally, we can extend the many-to-many model to network structures. For instance, unskilled workers may also want to go to training schools. The preferences of unskilled workers may satisfy the USC condition on the group of firms and the group of training schools if their selections on firms are more important for them. In this case, we have a ``(skilled workers)-(firms)-(unskilled workers)-(training schools)'' network. However, this example seems rather restricted. We leave for future research on economic applications under these extended frameworks.
\section{Appendix}

\subsection{Proof of Theorem \ref{thm_main}}

We first prove the following lemma.

\begin{lemma}\label{lma_NC}
\normalfont
For any $f\in F$, any $w\in W$, and any $X,X'\subseteq W$ with $w\in X$ and $w\notin X'$,

(i) if no worker from $X'$ is a complement to $w$ for $f$, then $w\in Ch_f(X\cup X')$ implies $w\in Ch_f(X)$;

(ii) if no worker from $X'$ is a substitute to $w$ for $f$, then $w\in Ch_f(X)$ implies $w\in Ch_f(X\cup X')$.

\end{lemma}

\begin{proof}
(i) Suppose for some $f,w,X$, and $X'$, $w\in Ch_f(X\cup X')$ and $w\notin Ch_f(X)$. Since no worker from $X'$ is a complement to $w$ for $f$, for any $w'\in X'$, we have $w\notin Ch_f(X\cup\{w'\})$. Then, for any $w''\in X'$ with $w''\neq w'$, we have $w\notin Ch_f(X\cup\{w'\}\cup\{w''\})$. Repeat this argument, we have $w\notin Ch_f(X\cup X')$. A contradiction.

(ii) Since no worker from $X'$ is a substitute to $w$ for $f$, for any $w'\in X'$, we have $w\in Ch_f(X\cup\{w'\})$. Then, for any $w''\in X'$ with $w''\neq w'$, we have $w\in Ch_f(X\cup\{w'\}\cup\{w''\})$. Repeat this argument, we have $w\in Ch_f(X\cup X')$.
\end{proof}

We define a \textbf{skilled-worker matching} by replacing $W$ by $S$ in Definition \ref{Def_matching}. The first stage of the mechanism produces a skilled-worker matching. Given a skilled-worker matching $\theta$, for each $f\in F$, let $A^S_f(\theta)\equiv\{s\in S\mid f\succeq_s \theta(s)\}$ be the available set of skilled workers for $f$ at $\theta$, i.e., the set of skilled workers who are matched with $f$, or with firms that are less preferred than $f$ according to their preferences. Analogously, given a matching $\eta$, for each $f\in F$ define $A_f(\eta)\equiv\{w\in W\mid f\succeq_w \eta(w)\}$ to be the available set of workers for $f$ at $\eta$. Note that $\theta(f)\subseteq A^S_f(\theta)$ and $\eta(f)\subseteq A_f(\eta)$.

Let $\theta$ and $\eta$ be the skilled-worker matching and the matching produced by the first stage and the two-stage DA, respectively. We now prove the following two lemmata.

\begin{lemma}\label{lma_stage1}
\normalfont
For each $f\in F$, $Ch_f(A^S_f(\theta))=\theta(f)$.
\end{lemma}

\begin{proof}
Fix a firm $f\in F$. $A^S_f(\theta)$ is the set of skilled workers who have proposed to $f$ in the first stage. For each $s$ such that $s\in A^S_f(\theta)$ and $s\notin\theta(f)$, $s$ is rejected by $f$ from a subset of $A^S_f(\theta)$ at some step; then because of assumption (i) of the USC condition and Lemma \ref{lma_NC}(i), we have $s\notin Ch_f(A^S_f(\theta))$. Hence, we have $Ch_f(A^S_f(\theta))\subseteq \theta(f)$. Since in the last step of the first stage $\theta(f)$ is chosen by $f$ from a subset of $A^S_f(\theta)$ that contains $\theta(f)$, we know that $f$ prefers $\theta(f)$ to any subset of $\theta(f)$. Therefore, we have $Ch_f(A^S_f(\theta))=\theta(f)$.
\end{proof}

\begin{lemma}\label{lma_stage2}
\normalfont
For each $f\in F$, $Ch_f(A_f(\eta))=\eta(f)$
\end{lemma}

\begin{proof}
Fix a firm $f\in F$. $A_f(\eta)$ is the set of workers who have proposed to $f$ during two stages. At the last step of the first stage, $\theta(f)$ is chosen by $f$, and thus we also have $\theta(f)=Ch_f(\theta(f))$. Then, at Step 2.1, because no unskilled worker is a substitute to any skilled worker, by Lemma \ref{lma_NC}(ii) we know that for any $s\in S$, $s\in \theta(f)$ implies that $s$ is chosen by $f$ at Step 2.1. For similar reasons, we know that for any $s\in S$, if $s\in \theta(f)$, then $s$ is chosen by $f$ at each step of the second stage, and thus $s\in \theta(f)$ implies $s\in \eta(f)$.

Hence, for each $s\in S$ such that $s\in A_f(\eta)$ and $s\notin\eta(f)$, $s$ is rejected by $f$ at some step of the first stage from a subset of $A_f(\eta)$. Because no worker is a complement to $s$, by Lemma \ref{lma_NC}(i) we have $s\notin Ch_f(A_f(\eta))$.

For each $u\in U$ such that $u\in A_f(\eta)$ and $u\notin\eta(f)$, $u$ is rejected by $f$ at some step of the second stage from a subset of $A_f(\eta)$, say $X$, where $\theta(f)\subseteq X$. For any $s\in S$ such that $s\in A^S_f(\theta)$ and $s\notin\theta(f)$, by Lemma \ref{lma_stage1} we have $s\notin Ch_f(\theta(f)\cup \{s\})$. Then, by Lemma \ref{lma_NC}(i) we have $s\notin Ch_f(X\cup \{s\})$ since $\theta(f)\subseteq X$ and no worker is a complement to $s$. Then, $Ch_f(X\cup \{s\})\subseteq X$, and thus $u\notin Ch_f(X)$ implies $u\notin Ch_f(X\cup \{s\})$. For any $s'$ such that $s'\in A^S_f(\theta)$, $s'\notin\theta(f)$ and $s'\neq s$, by Lemma \ref{lma_stage1} we have $s'\notin Ch_f(\theta(f)\cup \{s\}\cup \{s'\})$. By Lemma \ref{lma_NC}(i) we have $s'\notin Ch_f(X\cup \{s\}\cup \{s'\})$ since $\theta(f)\subseteq X$ and no worker is a complement to $s'$. Then, $Ch_f(X\cup \{s\}\cup \{s'\})\subseteq X$, and thus $u\notin Ch_f(X)$ implies $u\notin Ch_f(X\cup \{s\}\cup \{s'\})$. Repeat this argument, we have $u\notin Ch_f(X\cup [A^S_f(\theta)\setminus\theta(f)])=Ch_f(X\cup A^S_f(\theta))$, where the equality is due to $\theta(f)\subseteq X$. Because $X\subseteq A_f(\eta)$ and $A_f(\eta)\setminus A^S_f(\theta)$ contains only unskilled workers, by Lemma \ref{lma_NC}(i) we have $u\notin Ch_f(A_f(\eta))$.

From the last two paragraphs, we have $Ch_f(A_f(\eta))\subseteq \eta(f)$. Since in the last step of the second stage $\eta(f)$ is chosen by $f$ from a subset of $A_f(\eta)$, we know that $f$ prefers $\eta(f)$ to any subset of $\eta(f)$. Therefore, we have $Ch_f(A_f(\eta))=\eta(f)$.
\end{proof}

The procedure of the mechanism guarantees the individual rationality of $\eta$. Suppose $f$ and $X\subseteq W$ form a blocking coalition. Since $f\succeq_w\eta(w)$ for all $w\in X$, we have $X\subseteq A_f(\eta)$. Then, $X\succ_f \eta(f)$ contradicts Lemma \ref{lma_stage2}. Therefore, there is no blocking coalition, and thus $\eta$ is stable.

\subsection{Proof of Theorem \ref{thm_school}}

(1) Assumption (i) of Definition \ref{Def_usc}.

For any $s,s'\in S$ and any $X\subseteq W$ with $\{s,s'\}\subseteq X$, such that $s'\notin Ch_f(X\setminus\{s\})$, by (\ref{sc1}) the number of students from $X\setminus\{s\}$ who have higher priority than $s'$ is not less than $q_f$. Thus, we have $s'\notin Ch_f(X)$ as well. Therefore, $s$ is not a complement to $s'$.

For any $u,u'\in U$ and any $X\subseteq W$ with $\{u,u'\}\subseteq X$, such that $u'\notin Ch_f(X\setminus\{u\})$, by (\ref{sc2}) we have either (a) $\mid (X\setminus\{u\})^{\rhd_f u'} \mid\geq q_f$ or (b) $1-\frac{\mid (X\setminus\{u\})\cap S\mid}{\mid (X\setminus\{u\})^{\rhd_f u'} \mid+1}>\alpha_f$. If (a) holds, then $\mid X^{\rhd_f u'} \mid\geq q_f$, and thus $u'\notin Ch_f(X)$. If (b) holds, since $\mid X^{\rhd_f u'} \mid\geq \mid (X\setminus\{u\})^{\rhd_f u'} \mid$, we have $1-\frac{\mid X\cap S\mid}{\mid X^{\rhd_f u'} \mid+1}>\alpha_f$, and thus $u'\notin Ch_f(X)$. Therefore, $u$ is not a complement to $u'$.
\\
\quad\\
(2) Assumption (ii) of Definition \ref{Def_usc}.

For any $s\in S$, $u\in U$, and any $X\subseteq W$ with $\{s,u\}\subseteq X$, such that $s\in Ch_f(X\setminus\{u\})$, the number of students from $X\setminus\{u\}$ who have higher priority than $s$ is less than $q_f$. Since $s\rhd_f u$, we have $s\in Ch_f(X)$ as well. Therefore, $u$ is not a substitute to $s$.

For any $s\in S$, $u\in U$, and any $X\subseteq W$ with $\{s,u\}\subseteq X$, such that $s\notin Ch_f(X\setminus\{u\})$, the number of students from $X\setminus\{u\}$ who have higher priority than $s$ is no less than $q_f$. Thus, we have $s\notin Ch_f(X)$ as well. Therefore, $u$ is not a complement to $s$.

\subsection{Proof of Theorem \ref{thm_bidirection}}
\begin{lemma}\label{lma}
\normalfont
For any $w\in W$, $\mathbf{p}$ and $\mathbf{p}'$ such that $p(w)>p'(w)$ and $p(\widetilde{w})=p'(\widetilde{w})$ for all $\widetilde{w}\neq w$,

(i) If $X\in D(\mathbf{p}')$ and $w\notin X$, then $X\in D(\mathbf{p})$.

(ii) If $X\in D(\mathbf{p})$ and $w\in X$, then $X\in D(\mathbf{p}')$.
\end{lemma}

\begin{proof}
(i) $X\in D(\mathbf{p}')$ indicates $v(X)-c(X,\mathbf{p}')\geq v(\widetilde{X})-c(\widetilde{X},\mathbf{p}')$ for all $\widetilde{X}\subseteq W$, then for any $\mathbf{p}$ such that $p(w)>p'(w)$ and $p(\widetilde{w})=p'(\widetilde{w})$ for all $\widetilde{w}\neq w$, $w\notin X$ implies $v(X)-c(X,\mathbf{p})=v(X)-c(X,\mathbf{p}')\geq v(\widetilde{X})-c(\widetilde{X},\mathbf{p}')\geq v(\widetilde{X})-c(\widetilde{X},\mathbf{p})$ for all $\widetilde{X}\subseteq W$, and thus we have $X\in D(\mathbf{p})$.

(ii) $X\in D(\mathbf{p})$ indicates $v(X)-c(X,\mathbf{p})\geq v(X')-c(X',\mathbf{p})$ for all $X'\subseteq W$. Since $w\in X$, we have $v(X)-c(X,\mathbf{p}')\geq v(X')-c(X',\mathbf{p}')$ for all $X'\subseteq W$.
\end{proof}

(1) Suppose $w$ is a substitute to $w'$ for the firm, then there exist $\mathbf{p},\mathbf{p}'\in \mathbb{R}^n$ with $p(w)>p'(w)$ and $p'(\widetilde{w})=p(\widetilde{w})$ for all $\widetilde{w}\neq w$ such that $w'$ is demanded at $\mathbf{p}$ but not demanded at $\mathbf{p}'$. For any $X\in D(\mathbf{p})$, we cannot have both $w\in X$ and $w'\in X$, otherwise, by Lemma \ref{lma}(ii), we know that $X\in D(\mathbf{p}')$, which contradicts that $w'$ is not demanded at $\mathbf{p}'$.

We then know that the firm must attain a strictly larger maximum to (\ref{max}) under $\mathbf{p}'$ than under $\mathbf{p}$. Otherwise, suppose the firm attains the same maximum to (\ref{max}) under $\mathbf{p}$ and $\mathbf{p}'$, then for each $X\in D(\mathbf{p})$ with $w'\in X$, since $w\notin X$, we would have $X\in D(\mathbf{p}')$. This also contradicts that $w'$ is not demanded at $\mathbf{p}'$.

For any $X'\in D(\mathbf{p}')$, we must have $w\in X'$. Otherwise, suppose there exists $X'\in D(\mathbf{p}')$ with $w\notin X'$, then by Lemma \ref{lma}(i), we know that $X'\in D(\mathbf{p})$. Because $X'$ belongs to both $D(\mathbf{p})$ and $D(\mathbf{p}')$, the firm attains the same maximum to (\ref{max}) under $\mathbf{p}$ and $\mathbf{p}'$. A contradiction.

For any $X\in D(\mathbf{p})$ with $w'\in X$ and $X'\in D(\mathbf{p}')$ (note that $w\notin X$ and $w\in X'$), we have $v(X)-c(X,\mathbf{p})\geq v(X')-c(X',\mathbf{p})$ and $v(X)-c(X,\mathbf{p}')<v(X')-c(X',\mathbf{p}')$. The last inequality is strict because $w'$ is not demanded at $\mathbf{p}'$ implies $X\notin D(\mathbf{p}')$.
Then there exists $\widehat{p}\in(p'(w),p(w)]$ such that $v(X')-c(X',\mathbf{\widehat{p}})=v(X)-c(X,\mathbf{p})=v(X)-c(X,\mathbf{\widehat{p}})$ where $\widehat{p}(w)=\widehat{p}$ and $\widehat{p}(\widetilde{w})=p(\widetilde{w})$ for all $\widetilde{w}\neq w$. The last equality is because $w\notin X$. We then show that $X,X'\in D(\mathbf{\widehat{p}})$. Suppose instead $X,X'\notin D(\mathbf{\widehat{p}})$ and $X''\in D(\mathbf{\widehat{p}})$, then (i) if $w\in X''$, $v(X'')-c(X'',\mathbf{\widehat{p}})>v(X')-c(X',\mathbf{\widehat{p}})$ and $w\in X'$ implies $v(X'')-c(X'',\mathbf{p}')>v(X')-c(X',\mathbf{p}')$, which contradicts $X'\in D(\mathbf{p}')$; (ii) if $w\notin X''$, $v(X'')-c(X'',\mathbf{\widehat{p}})>v(X)-c(X,\mathbf{\widehat{p}})$ and $w\notin X$ implies $v(X'')-c(X'',\mathbf{p})>v(X)-c(X,\mathbf{p})$, which contradicts $X\in D(\mathbf{p})$. Finally, we know for all $\delta>0$ $w$ is not demanded at $\mathbf{\widehat{p}}'$, where $\widehat{p}'(w')=\widehat{p}(w')-\delta$ and $\widehat{p}'(\widetilde{w})=\widehat{p}(\widetilde{w})$ for all $\widetilde{w}\neq w'$. This is because for any $X''\subseteq W$ with $w\in X''$,

(i) if $w'\notin X''$, we have

\begin{align*}
v(X'')-c(X'',\mathbf{\widehat{p}}')&=v(X'')-c(X'',\mathbf{\widehat{p}})\\
&\leq v(X)-c(X,\mathbf{\widehat{p}})\\
&<v(X)-c(X,\mathbf{\widehat{p}}')
\end{align*}
The first inequality is due to $X\in D(\mathbf{\widehat{p}})$. The second inequality is because $w'\in X$. Thus, we have $X''\notin D(\mathbf{\widehat{p}}')$.

(ii) if $w'\in X''$, we have $X''\notin D(\mathbf{p}')$, and then

\begin{align*}
v(X'')-c(X'',\mathbf{\widehat{p}}')&=v(X'')-c(X'',\mathbf{\widehat{p}})+\delta\\
&= v(X'')-c(X'',\mathbf{p}')+p'(w)-\widehat{p}+\delta\\
&<v(X')-c(X',\mathbf{p}')+p'(w)-\widehat{p}+\delta
\end{align*}

The inequality is because $X''\notin D(\mathbf{p}')$ and $X'\in D(\mathbf{p}')$. We then have
\begin{align*}
v(X')-c(X',\mathbf{p}')+p'(w)-\widehat{p}+\delta&=v(X')-c(X',\mathbf{\widehat{p}})+\delta\\
&=v(X)-c(X,\mathbf{\widehat{p}})+\delta\\
&=v(X)-c(X,\mathbf{\widehat{p}}')
\end{align*}

The second equality is because $X,X'\in D(\mathbf{\widehat{p}})$. Therefore, we have $v(X'')-c(X'',\mathbf{\widehat{p}}')<v(X)-c(X,\mathbf{\widehat{p}}')$, which indicates $X''\notin D(\mathbf{\widehat{p}}')$.

We have done since $w$ is demanded at $\mathbf{\widehat{p}}$, but not demanded at $\mathbf{\widehat{p}}'$.

(2) Suppose $w$ is a complement to $w'$ for the firm, then there exist $\mathbf{p},\mathbf{p}'\in \mathbb{R}^n$ with $p(w)>p'(w)$ and $p(\widetilde{w})=p'(\widetilde{w})$ for all $\widetilde{w}\neq w$ such that $w'$ is not demanded at $\mathbf{p}$ but demanded at $\mathbf{p}'$. For any $X'\in D(\mathbf{p}')$ with $w'\in X'$, we have $w\in X'$, otherwise by Lemma \ref{lma}(i) we have $X'\in D(\mathbf{p})$, which contradicts that $w'$ is not demanded at $\mathbf{p}$.

We also know that $w$ is not demanded at $\mathbf{p}$. Otherwise, suppose there exists $X\in D(\mathbf{p})$ with $w\in X$, then by Lemma \ref{lma}(ii) $X\in D(\mathbf{p}')$. For any $X'\in D(\mathbf{p}')$ with $w,w'\in X'$, we have $v(X')-c(X',\mathbf{p}')=v(X)-c(X,\mathbf{p}')$, and then $v(X')-c(X',\mathbf{p}')+p'(w)-p(w)=v(X)-c(X,\mathbf{p}')+p'(w)-p(w)$, which is $v(X')-c(X',\mathbf{p})=v(X)-c(X,\mathbf{p})$ since $w\in X'$ and $w\in X$. Then $X\in D(\mathbf{p})$ implies $X'\in D(\mathbf{p})$, which contradicts that $w'$ is not demanded at $\mathbf{p}$.

For any $X'\in D(\mathbf{p}')$ with $w,w'\in X'$ and for any $X\in D(\mathbf{p})$ (note that $w\notin X$ and $w'\notin X$), we have $v(X')-c(X',\mathbf{p}')\geq v(X)-c(X,\mathbf{p}')$, and $v(X')-c(X',\mathbf{p})< v(X)-c(X,\mathbf{p})$. The last inequality is strict because $w$ and $w'$ are not demanded at $\mathbf{p}$ implies $X'\notin D(\mathbf{p})$. Then, there exists $\widehat{p}\in[p'(w),p(w))$ such that $v(X')-c(X',\mathbf{\widehat{p}})=v(X)-c(X,\mathbf{\widehat{p}})$ where $\widehat{p}(w)=\widehat{p}$ and $\widehat{p}(\widetilde{w})=p(\widetilde{w})$ for all $\widetilde{w}\neq w$. We would have $X,X'\in D(\mathbf{\widehat{p}})$, otherwise suppose $X''\in D(\mathbf{\widehat{p}})$ with $v(X'')-c(X'',\mathbf{\widehat{p}})>v(X)-c(X,\mathbf{\widehat{p}})$$=v(X')-c(X',\mathbf{\widehat{p}})$, then (i) if $w\notin X''$, then  $v(X'')-c(X'',\mathbf{\widehat{p}})>v(X)-c(X,\mathbf{\widehat{p}})$ implies $v(X'')-c(X'',\mathbf{p})>v(X)-c(X,\mathbf{p})$, which contradicts $X\in D(\mathbf{p})$; (ii) if $w\in X''$, then  $v(X'')-c(X'',\mathbf{\widehat{p}})>v(X')-c(X',\mathbf{\widehat{p}}')$ implies $v(X'')-c(X'',\mathbf{p}')>v(X')-c(X',\mathbf{p}')$, which contradicts $X'\in D(\mathbf{p}')$.

Let $\gamma=\min_{w\notin Y,w'\in Y} v(X)-c(X,\mathbf{\widehat{p}})-[v(Y)-c(Y,\mathbf{\widehat{p}})]$, then for all $\delta\in(0,\gamma]$, we have $X,X'\in D(\mathbf{p}^*)$, where $p^*(w)=\widehat{p}(w)+\delta$, $p^*(w')=\widehat{p}(w')-\delta$, and $p^*(\widetilde{w})=\widehat{p}(\widetilde{w})$ for all $\widetilde{w}\neq w,w'$.\footnote{We know that $Y\notin D(\mathbf{\widehat{p}})$ for any $Y\subseteq W$ with $w\notin Y$ and $w'\in Y$, otherwise by Lemma \ref{lma}(i) $Y\in D(\mathbf{p})$, which contradicts that $w'$ is not demanded at $\mathbf{p}$. Thus, we have $\gamma>0$.} This is because for any $Y\subseteq W$ with $w\notin Y$ and $w'\in Y$,
\begin{align*}
v(Y)-c(Y,\mathbf{p}^{*})&=v(Y)-c(Y,\mathbf{\widehat{p}})+\delta\\
&\leq v(Y)-c(Y,\mathbf{\widehat{p}})+\gamma\leq v(X)-c(X,\mathbf{\widehat{p}})=v(X)-c(X,\mathbf{p}^{*})=v(X')-c(X',\mathbf{p}^{*}),
\end{align*}
where the last inequality is by the definition of $\gamma$, the second equality is due to $w,w'\notin X$, and the last equality is because $X,X'\in D(\mathbf{\widehat{p}})$.

Finally, we want to show that $w$ is not demanded at $\mathbf{p}^{**}$, where $p^{**}(w)=\widehat{p}(w)+\delta$ and $p^{**}(\widetilde{w})=\widehat{p}(\widetilde{w})$ for all $\widetilde{w}\neq w$. By Lemma \ref{lma}(i), $X\in D(\mathbf{p}^*)$ implies $X\in D(\mathbf{p}^{**})$. For any $Y'$ with $w\in Y'$,
\begin{align*}
v(Y')-c(Y',\mathbf{p}^{**})&=v(Y')-c(Y',\mathbf{\widehat{p}})-\delta\\
&\leq v(X)-c(X,\mathbf{\widehat{p}})-\delta=v(X)-c(X,\mathbf{p}^{**})-\delta<v(X)-c(X,\mathbf{p}^{**})
\end{align*}
The first inequality is because $X\in D(\mathbf{\widehat{p}})$, the second equality is because $w\notin X$. Hence, we have $Y'\notin D(\mathbf{p}^{**})$. Therefore, $w$ is not demanded at $\mathbf{p}^{**}$, but demanded at $\mathbf{p}^{*}$ ($X'\in D(\mathbf{p}^{*})$).


\bigskip


\begin{thebibliography}{99999999999999999999999999999999999999999}

\bibitem[Ashlagi, Braverman, and Hassidim (2014)]{ABH14}{\small Ashlagi, I., Braverman, M., Hassidim, A., 2014. Stability in large matching markets with complementarities. \emph{Operations Research}, 62, 713-732.}

\bibitem[Azevedo, Weyl, and White(2013)]{AWW13}{\small Azevedo, E.M., Weyl, E.G., White, A., 2013. Walrasian equilibrium in large, quasi-linear markets. \emph{Theoretical Economics}, 8(2), 281-290. }

\bibitem[Azevedo and Hatfield(2018)]{AH18}{\small Azevedo, E.M., Hatfield, J.M., 2018. Existence of equilibrium in large matching markets with complementarities. \emph{working paper}.}

\bibitem[Baldwin and Klemperer(2019)]{BK19}{\small Baldwin, E., Klemperer, P., 2019. Understanding Preferences: ``Demand Types'', and the Existence of Equilibrium with Indivisibilities. \emph{Econometrica}, 87, 867-932. }

\bibitem[Che, Kim, and Kojima(2019)]{CKK19}{\small Che, Y-K., Kim, J., Kojima, F., 2019. Stable matching in large economies. \textit{Econometrica}, 87(1), 65-110.}

\bibitem[Danilov, Koshevoy, and Murota(2001)]{DKM01}{\small Danilov, V., Koshevoy, G., Murota, K., 2001. Discrete convexity and equilibria in economics with indivisible goods and money. \emph{Mathematical Social Sciences}, 41, 251-273.}

\bibitem[Echenique and Oviedo(2006)]{EO06}{\small Echenique, F., Oviedo, J., 2006. A theory of stability in many-to-many matching. \textit{Theoretical Economics}, 1, 233-273.}

\bibitem[Echenique and Yenmez(2007)]{EY07}{\small Echenique, F., Yenmez, B., 2007. A Solution to Matching With Preferences Over Colleagues. \textit{Games and Economic Behavior}, 59, 46-71.}

\bibitem[Echenique and Yenmez(2015)]{EY15}{\small Echenique, F., Yenmez, B., 2015. How to control controlled school choice. \textit{American Economic Review}, 105(8), 2679-2694.}

\bibitem[Ehlers, Hafalir, Yenmez, and Yildirim(2014)]{EHYY14}{\small Ehlers, L., Hafalir, I.E., Yenmez, M.B., Yildirim, M.A., 2014. School choice with controlled choice constraints: Hard bounds vs. soft bounds. \textit{Journal of Economic Theory}, 153, 648-683.}

\bibitem[Gale and Shapley(1962)]{GS62}{\small Gale, D., Shapley, L.S., 1962. College admissions and the stability of marriage. \textit{American Mathematical Monthly}, 69, 9-15.}

\bibitem[Hafalir, Kojima, and Yenmez(2019)]{HKY19}{\small Hafalir, I. E., Kojima, F.,  Yenmez, M.B., 2019. Interdistrict school choice: A theory of student assignment. \emph{working paper}.}

\bibitem[Hatfield, Kominers, Nichifor, Ostrovsky, and Westkamp(2013)]{HKNOW13}{\small Hatfield, J. W., Kominers, S. D.,  Nichifor, A., Ostrovsky, M., Westkamp, A., 2013. Stability and competitive equilibrium in trading networks. \textit{Journal of Political Economy}, 121(5), 966-1005.}

\bibitem[Hatfield and Milgrom(2005)]{HM05}{\small Hatfield, J.W., Milgrom, P.R., 2005. Matching with contracts. \textit{American Economic Review}, 95, 913-935.}

\bibitem[Huang(2021)]{H21}{\small Huang, C., 2021. Stable matching: An integer programming approach. \emph{working paper}, https://arxiv.org/abs/2103.03418}

\bibitem[Kelso and Crawford(1982)]{KC82}{\small Kelso, A. S., Crawford, V.P., 1982. Job matching, coalition formation and gross substitutes. \textit{Econometrica}, 50, 1483-1504. }

\bibitem[Klaus and Klijn(2005)]{KK05}{\small Klaus, B., Klijn, F., 2005. Stable Matchings and Preferences of Couples. \textit{Journal of Economic Theory}, 121, 75-106.}

\bibitem[Kojima, Pathak, and Roth(2013)]{KPR13}{\small Kojima, F., Pathak, P.A., Roth, A.E., 2013. Matching with Couples: Stability and Incentives in Large Markets. \textit{Quarterly Journal of Economics}, 128, 1585-1632.}

\bibitem[Nguyen and Vohra(2018)]{NV18}{\small Nguyen, T., Vohra, R., 2018. Near-Feasible Stable Matchings with Couples. \textit{American Economic Review}, 108(11), 3154-3169.}

\bibitem[Nguyen and Vohra(2019)]{NV19}{\small Nguyen, T., Vohra, R., 2019. Stable Matching with Proportionality Constraints. \textit{Operations Research}, 67(6), 1503-1519.}

\bibitem[Ostrovsky(2008)]{O08}{\small Ostrovsky, M., 2008. Stability in supply chain networks. \textit{American Economic Review}, 98, 897-923. }

\bibitem[Pycia(2012)]{P12}{\small Pycia, M., 2012. Stability and preference alignment in matching and coalition formation. \textit{Econometrica}, 80(1), 323-362. }

\bibitem[Roth(1984)]{R84}{\small Roth, A.E., 1984. Stability and Polarization of Interests in Job Matching. \emph{Econometrica}, 52, 47-57.}

\bibitem[Roth(1985)]{R85}{\small Roth, A.E., 1985. Conflict and coincidence of interest in job matching: Some new results and open questions. \emph{Mathematics of Operations Research}, 10, 379-389.}

\bibitem[Roth and Sotomayor(1990)]{RS90}{\small Roth, A.E., Sotomayor, M., 1990. Two-sided Matching: A Study in Game-Theoretic Modelling and Analysis. Econometric Society Monographs No. 18, Cambridge University Press, Cambridge England.}

\bibitem[Sotomayor(1999)]{S99}{\small Sotomayor, M., 1999. Three remarks on the many-to-many stable matching problem. \textit{Mathematical Social Sciences}, 38, 55-70.}

\bibitem[Sun and Yang(2006)]{SY06}{\small Sun, N., Yang, Z., 2006. Equilibria and indivisibilities: Gross substitutes and complements. \textit{Econometrica}, 74, 1385-1402.}

\bibitem[Sun and Yang(2009)]{SY09}{\small Sun, N., Yang, Z., 2009. A double-track adjustment process for discrete markets with substitutes and complements. \textit{Econometrica}, 77, 993-952.}
\end{thebibliography}
\end{document}